\documentclass{article}
\usepackage{amsmath}
\usepackage{amsfonts}
\usepackage{amssymb}
\usepackage{amsthm}

\def\l{\lambda}
\def\p{\partial}
\def\e{\mathrm{e}}
\newtheorem{theorem}{Theorem}
\newtheorem{prop}{Proposition}

\newtheorem*{corollary}{Corollary}
\newcommand{\dbar}{\bar{\partial}}

\newcommand{\be}{\begin{equation}}
\newcommand{\ee}{\end{equation}}
\newcommand{\bea}{\begin{eqnarray}}
\newcommand{\eea}{\end{eqnarray}}
\newcommand{\beaa}{\begin{eqnarray*}}
\newcommand{\eeaa}{\end{eqnarray*}}

\newcommand{\nn}{\nonumber}
\renewcommand{\d}{\mathrm{d}}

\usepackage{authblk}
\begin{document}
\title{Matrix extension of the 
Manakov-Santini system
and integrable
chiral model on Einstein-Weyl background}
\author{L.V. Bogdanov \thanks{leonid@itp.ac.ru}}
\affil{L.D. Landau ITP RAS, Moscow, Russia}
\date{}
\maketitle
\begin{abstract} 
It was demonstrated recently \cite{DFK14} 
that
the Manakov-Santini system describes a  
local form of general 
Lorentzian
Einstein-Weyl geometry.
We introduce integrable matrix extension of the
Manakov-Santini system and show that it 
describes (2+1)-dimensional integrable chiral model
in Einstein-Weyl space. We develop 
a dressing scheme for the extended MS system
and define an extended hierarchy. Matrix
extension of  Toda type
system connected with another local form
of Einstein-Weyl geometry 
is also considered.
\end{abstract}
\section{Introduction}
It was demonstrated recently 
\cite{DFK14} 
that local form of important integrable 
geometric structures
connected with conformal self-duality
is described in general case by  dispersionless
integrable systems (twistor 
integrability of these geometric
structures was established in \cite{Penrose}, \cite{Atiah}).
A crucial observation made in 
\cite{DFK14} is that
any real conformal anti-self-dual (ASD) structure 
in signature (2,2) (or general complex 
analytic conformal ASD structure)
can be locally
represented
by a metric
\bea
\label{ASDmetric1}
{\textstyle\frac{1}{2}}g=dwdx-dzdy-F_y dw^2-(F_x-G_y)dwdz + G_xdz^2,
\eea
where the functions $F$, $G$ satisfy the 
compatibility conditions $[X_1,X_2]=0$
for the 
dispersionless
Lax pair represented by vector fields
containing the derivative over
spectral parameter $\lambda$
(introduced in \cite{BDM07} in the context of
dispersionless integrable hierarchies, see also \cite{LVB11})
\bea
X_1=\p_z-\lambda\p_x
+F_x\p_x+G_x\p_y+ f_1\p_\lambda,
\nn
\\
X_2=\p_w- \lambda\p_y + F_y\p_x+G_y\p_y+f_2\p_\lambda,
\label{new_lax}
\eea
which can be written as a coupled system of third order
PDEs
\bea
&&
\p_x(Q(F))+\p_y(Q(G))=0,
\nn\\
&&
(\p_w+F_y\p_x+G_y\p_y)Q(G)+(\p_z+F_x\p_x+G_x\p_y)Q(F)=0,
\label{sd_3rd}
\eea
where
\[
Q=\p_w\p_x-\p_z\p_y+F_y{\p_x}^2-G_x{\p_y}^2-
(F_x-G_y)\p_x\p_y,
\]
and, 
due to compatibility conditions, $f_1$ and $f_2$ are 
expressed through $F$ and $G$,
\beaa
&&
f_1=-Q(G),\quad f_2=Q(F).
\eeaa

In \cite{LVB17} we have introduced a matrix extension of
Lax pair (\ref{new_lax}),
\be
\begin{split}
&
\nabla_{X_1}=\p_z-\lambda\p_x
+F_x\p_x+G_x\p_y+ f_1\p_\lambda + A_1,
\\
&
\nabla_{X_2}=\p_w- \lambda\p_y + F_y\p_x+G_y\p_y+f_2\p_\lambda + A_2,
\end{split}
\label{extLax}
\ee 
where gauge field components $A_1$, $A_2$ do not 
depend on $\lambda$ and take their values in some (matrix)
Lie algebra.
Lax pairs of this structure 
(without derivative over $\lambda$) were already present 
in the seminal work of
Zakharov and Shabat \cite{ZS} (1979). The commutation relation splits into (scalar) vector
field part, which is the same as for unextended Lax pair,
and matrix part
\be
\begin{split}
&
\partial_x A_2=\partial_y A_1,
\\
&
(\p_z
+F_x\p_x+G_x\p_y)A_2 - 
(\p_w+ F_y\p_x+G_y\p_y)A_1 + [A_1,A_2]=0,
\end{split}
\label{SDYM}
\ee
which contains the coefficients of the metric
as some kind of `background'. 

The following
statement proved in \cite{LVB17}
demonstrates that scalar background 
in equations (\ref{SDYM})
has a direct geometric sense.
\begin{prop} 
\label{prop_ASDYM}
Equations (\ref{SDYM}) represent 
anti-self-dual Yang-Mills
(ASDYM) equations 
for the background conformal 
structure (\ref{ASDmetric1}) 
(in a special gauge).
\end{prop}
Different reductions of ASDYM equations give rise to
integrable background geometries which are connected to
dispersionless integrable systems. The current picture
of the field and many examples are provided in  
\cite{Calderbank}.

In (2+1)-dimensional case conformal 
ASD structure reduces to Einstein-Weyl 
geometry, and it was proved in \cite{DFK14}
that for Lorentzian signature
it is locally described by the Manakov-Santini 
system \cite{MS06,MS07}
\be
\begin{split}
u_{xt} &= u_{yy}+(uu_x)_x+v_xu_{xy}-u_{xx}v_y,
\\
v_{xt} &= v_{yy}+uv_{xx}+v_xv_{xy}-v_{xx}v_y
\end{split}
\label{MSeq}
\ee
with the Lax pair 
\be
\begin{split}
X_1&=\partial_y-(\l-v_{x})\partial_x + u_{x}\partial_\lambda,
\\
X_2&=\partial_t-((\l^2-v_{x}\l+u -v_{y})\partial_x
+(u_{x}\l+u_{y})\partial_\lambda
\end{split}
\label{MSLax}
\ee
In \cite{LVB17} we presented matrix 
extension of the Lax pair
(\ref{MSLax}),
\be
\begin{split}
&
\nabla_{X_1}=
\partial_y-(\lambda-v_{x})\partial_x + u_{x}\partial_\lambda + A,
\\
&
\nabla_{X_2}=\partial_t-(\lambda^2-v_{x}\lambda + u -v_{y})\partial_x
+(u_{x}\lambda+u_{y})\partial_\lambda +\lambda A + B,
\end{split}
\label{MSLaxext}
\ee
leading to the matrix system on the
background of the Manakov-Santini system
\be
\begin{split}
&
\p_y A - \p_x B=0,
\\
&
(\p_y+v_x\p_x)B-(\p_t+(v_y-u)\p_x) A + u_x A +[A,B]=0
\end{split}
\label{MSext}
\ee
For the potential $K$, $A=K_x$, $B=K_y$
we  have
\bea
K_{tx}-K_{yy} - [K_x,K_y] - \p_x(u K_x) + v_y K_{xx} - v_x K_{xy}=0,
\label{mono}
\eea
where $u$, $v$ satisfy Manakov-Santini system describing
Einstein-Weyl geometry. For trivial background $u=v=0$
this equation represents one of the forms of
integrable chiral model \cite{Ward}, \cite{Dun}.

In Section \ref{Chir} we shall demonstrate
that
{\em system (\ref{MSext}) represents the 
first-order Yang-Mills-Higgs
system introduced
by Ward \cite{Ward}, defining integrable 
chiral model, on Einstein-Weyl background in the form 
described by the Manakov-Santini system}. 

In Section \ref{Ext} we develop 
the techique of matrix extension of dispersionless
integrable hierarchies introduced in \cite{LVB17}
for the case of the Manakov-Santini hierarchy,
formulate the dressing scheme and derive Lax-Sato equations
for the extended hierarchy.

Another general local form of Lorentzian Einstein-Weyl
structure \cite{DFK14}
is given by 
two-component generalization of the dispersionless 2DTL equation \cite{LVB10}
\be
\begin{split}
(\mathrm{e}^{-\phi})_{tt}&=m_t\phi_{xy}-m_x\phi_{ty},
\\
m_{tt}\mathrm{e}^{-\phi}&=m_{ty}m_x-m_{xy}m_t
\end{split}
\label{gen2DTL}
\ee
with the Lax pair 
\be
\begin{split}
X_1 &=\partial_x-
\bigl(
\lambda+ \frac{m_x}{m_t}
\bigr)
\partial_t + 
\lambda
\bigl(
\phi_t \frac{m_x}{m_t}-\phi_x
\bigr)
\partial_\lambda,
\\
X_2 &=\partial_y-\frac{1}{\lambda}\frac{\mathrm{e}^{-\phi}}{m_t}\partial_t -
\frac{(\mathrm{e}^{-\phi})_t}{m_t}\partial_\lambda.
\end{split}
\label{Laxgen2DTL}
\ee
For $m=t$ the system (\ref{gen2DTL}) reduces to the dispersionless 2DTL equation
\bea
(\mathrm{e}^{-\phi})_{tt}=\phi_{xy},
\label{d2DTL}
\eea
System (\ref{gen2DTL}) doesn't preserve the symmetry of the dispersionless 2DTL equation 
with respect 
to $x$, $y$ variables, however, it is possible 
to introduce symmetric generalizations of
the d2DTL equation, including elliptic case with
complex variables $z$, $\bar z$ instead of $x$, $y$
\cite{LVB10}.
It is an interesting question whether it is possible
to represent Einstein-Weyl geometry in Euclidean
signature in this way.

Lax pair (\ref{Laxgen2DTL}) posesses
a natural matrix extension
\be
\begin{split}
\nabla_{X_1}&=\partial_x-(\lambda+ \frac{m_x}{m_t})\partial_t + 
\lambda(\phi_t \frac{m_x}{m_t}-\phi_x)\partial_\lambda
+A
\\
\nabla_{X_2}&=\partial_y-\frac{1}{\lambda}\frac{\mathrm{e}^{-\phi}}{m_t}\partial_t -
\frac{(\mathrm{e}^{-\phi})_t}{m_t}\partial_\lambda
+\frac{1}{\lambda}B,
\end{split}
\ee
leading to matrix equations
\be
\begin{split}
&\p_t B +\p_y A=0,
\\
&(\partial_x-\frac{m_x}{m_t}\partial_t)B
+ \frac{\mathrm{e}^{-\phi}}{m_t}\partial_t A
-(\phi_t \frac{m_x}{m_t}-\phi_x)B + [A,B]=0,
\end{split}
\label{TodaM} 
\ee
or, in terms of potential $K$,
$A=K_t$, $B=-K_y$,
\be
\e^{-\phi}K_{tt} + {m_x}K_{yt} - {m_t}K_{xy}
+(\phi_t {m_x}-\phi_x {m_t})K_y +{m_t}[K_y,K_t]=0.
\ee
This equation on trivial background $m=t$, $\phi=1$
coincides (up to a change of variables) with equation
(\ref{mono}) on trivial background,
representing integrable chiral model \cite{Dun}.

It is a natural conjecture that system (\ref{TodaM}),
similar to system (\ref{MSext}), represents
an integrable chiral model on Einstein-Weyl background.
Though we will not provide a proof in the present work,
it should be completely analogous to the case of 
Manakov-Santini system. In Section \ref{Ext1} we develop 
the techique of matrix extension of dispersionless
integrable hierarchies 
for the case of Toda-type hierarchy connected
with the system (\ref{gen2DTL}),
formulate the dressing scheme and derive Lax-Sato equations
for the extended hierarchy.
\section{Integrable chiral model on the background
of Einstein-Weyl geometry}
\label{Chir}
Let us recall the results relating Einstein-Weyl geometry
and the Manakov-Santini system \cite{DFK14}.

EW geometry on a three-dimensional manifold $M^3$ consists of a conformal structure $[g]$
and a symmetric connection
${\mathbb{D}}$ compatible with $[g]$ in the sense that, 
for any $g\in [g]$, \ 
\[
{\mathbb {D}}g=\omega\otimes g
\]
for some covector $\omega$, and such that
the trace-free part of the symmetrized Ricci tensor of ${\mathbb{D}}$
vanishes. 
\begin{theorem}[Dunajski, Ferapontov and Kruglikov (2014)]
\label{theo_ms}
There exists a local coordinate system $(x, y, t)$ on $M^3$ such that
any Lorentzian Einstein-Weyl structure is locally of the form
\begin{gather}
\begin{split}
g &= -(dy + v_x dt )^2 +4(dx + (u - v_y ) dt ) dt,
\\
\omega &= v_{xx} dy+(-4u_x + 2v_{xy} +v_xv_{xx})dt,
\end{split}
\label{metricMS}
\end{gather}
where the functions  $u$ and $v$ on $M^3$ satisfy a coupled system of second-order PDEs
\begin{equation}
P(u) - u_x^2=0, \quad P(v)=0,
\label{MS}
\end{equation}
where 
\be
P=\p_x\p_t-\p_y^2+(v_y-u){\p_x}^2 - v_x{\p_x\p_y}. 
\label{metricMS1}
\ee
\end{theorem}
System (\ref{MS}) coincides with  Manakov-Santini system (\ref{MSeq}).
It was shown in \cite{Dun3} that any solution to 
(\ref{MS}) gives rise to an EW structure of the form (\ref{metricMS}),
but the question whether all EW structures arise in that way has remained 
open.

Let us introduce
a gauge field (potential) ${A}$, which is a one-form
taking its values in some (matrix) Lie algebra, and
the two-form
${F}=\d{A}+{A}\wedge{A}$
(connection curvature, field intensity) and consider 
the equation
\bea
D\Phi+\tfrac{1}{2}\omega \Phi=*{F},
\label{Ward}
\eea
where $D\Phi=\d\Phi+[{A},\Phi]$, $\Phi$ is a function
taking values in the Lie algebra (Higgs field, \cite{Dun}).
This equation for Minkowski metric coincides the
Yang-Mills-Higgs system
introduced by Ward \cite{Ward}, \cite{Dun}, leading to integrable
chiral model. Equation (\ref{Ward}) 
represents an integrable Ward system
on Einstein-Weyl geometry background,
the term  $\tfrac{1}{2}\omega \Phi$ is responsible
for correct behavior of equation (\ref{Ward})
under conformal transformation $g\rightarrow f g$,
$\Phi\rightarrow f^{-\frac{1}{2}}\Phi$.
\begin{prop} There exists a gauge for which 
equation (\ref{Ward}) for Einstein-Weyl structure
(\ref{metricMS})
takes the form
of matrix extension of MS system (\ref{MSext}), (\ref{mono}).
For this gauge $\Phi=-A_y$, $A_x=0$, $A=2A_y$,
$B=A_t-v_x A_y$.
\end{prop}
\begin{proof}
This proposition is proved directly, introducing
a basis (frame, vielbein) of vector fields and dual basis
of forms in which the metric takes a simple form
and writing down components of equation (\ref{Ward})
with respect to this basis, using a standard formula
\beaa
{F}({u},{v})=\nabla_{u}\nabla_{v}
-\nabla_{v}\nabla_{u}-
\nabla_{[{u},{v}]}
\eeaa
valid for arbitrary vector fields ${u}$, ${v}$.

Let us introduce a basis of vector fields
\begin{gather}
\begin{split}
\mathbf{e}_1&=\p_x,
\\
\mathbf{e}_2&=\p_y,
\\
\mathbf{e}_3&=\p_t + (v_y-u)\p_x -v_x\p_y,
\end{split}
\end{gather}
and dual basis of forms
\begin{gather}
\begin{split}
\mathbf{e}^1&=dx - (v_y-u)dt,
\\
\mathbf{e}^2&=dy + v_x dt,
\\
\mathbf{e}^3&=\p_t ,
\end{split}
\end{gather}
Then metric (\ref{metricMS}) takes the form
\bea
g=-(\mathbf{e}^2)^2+4\mathbf{e}^1 \mathbf{e}^3,
\eea
and symmetric bivector (\ref{metricMS}) (inverse metric) is
\bea
P=-(\mathbf{e}_2)^2+\mathbf{e}_1 \mathbf{e}_3.
\eea
First component of equation (\ref{Ward})
gives
\bea
[\p_x + A_x,\p_y + A_y + \Phi]=0,
\eea
that implies the existence of gauge in which
$A_x=0$, $A_y+\Phi=0$.

Second component reads
\beaa
-\p_y A_y - \tfrac{1}{2}v_{xx} A_y=\tfrac{1}{2}(-\p_x A_3
-v_{xx} A_y),
\eeaa
where $A_3=A_t-v_y A_y$, corresponding
to first equation of system (\ref{MSext}) and
leading to the existence of potential $K$,
$\p_x K=2 A_y$, $\p_y K= A_3.$

Third component 
\beaa
(2u_x -v_{xy})A_y
=
2(\p_t + (v_y-u)\p_x -v_x\p_y) A_y -\p_y A_3
-v_{yx}A_y + 2[A_3,A_y]
\eeaa
after some transformations takes the form 
of second equation of system (\ref{MSext})
\end{proof}
\begin{corollary}
There exist local coordinates and a gauge such that
for any Lorentzian Einstein-Weyl structure 
equation (\ref{Ward}) is locally of the form
(\ref{MSext}), (\ref{mono}).
\end{corollary}

\section{Manakov-Santini hierarchy and its
matrix extension}
\label{Ext}

The crucial object for matrix extension of involutive
distribution of
polynomial (meromorphic) in spectral parameter vector fields with
the basis $X_i$ is the matrix function $\Psi$ possessing the 
property that all the functions $\Psi^{-1}X_i\Psi$ are polynomial (meromorphic). The function $\Psi$
is suggested to be bounded and without zeroes in the
spectral plane and analytic in some neighborhood
of poles of vector fields.
Extended linear problems then read
\be
X_i\Psi=(\Psi^{-1}X_i\Psi)_+\Psi
\ee
where $(\Psi^{-1}X_i\Psi)_+$ is polynomial part
of the function represented as Laurent series
containing finite number of matrix fields
(meromorphic part in the case of multiple poles).
This function can be constructed using a matrix Riemann-Hilbert (RH) problem
of the form
\bea
\Psi_{+}=\Psi_{-}R(\psi_1,\dots,\psi_N),
\label{RH1}
\eea
defined on some oriented curve $\gamma$ in the complex plane,
or matrix $\dbar$ problem
\bea
\dbar\Psi=\Psi R(\psi_1,\dots,\psi_N),
\label{dbar1}
\eea
defined in some region $G$, where $\psi_k(\lambda,\mathbf{t})$ 
are wave functions for the distribution,
\be
X_i\psi_k=0 \quad \forall i,k
\ee
defined on $\gamma$ or in $G$. 
We suggest that solution $\Psi$ of RH (or $\dbar$) problem
is bounded
and has no zeroes, normalisation by 1 at infinity fixes the gauge
and leads to closed systems of equations for matrix coefficients
of extended Lax pairs. The dispersionless hierarchy corresponding
to integrable distribution with the basis $X_i$ plays the role of
the background.

Manakov-Santini hierarchy is defined by Lax-Sato equations \cite{LVB09}
\bea
&&
\frac{\partial}{\partial t_n}
\begin{pmatrix}
L\\
M
\end{pmatrix}
=
\left(\left(
\frac{ L^n L_p}{\{L,M\}}\right)_+
{\partial_x}
-\left(\frac{ L^n L_x}{\{L,M\}}\right)_+
{\partial_\lambda}\right)\begin{pmatrix}
L\\
M
\end{pmatrix},
\label{genSato1}
\eea 
where $L$, $M$, corresponding to Lax and Orlov functions of dispersionless KP hierarchy,
are the series
\bea
&&
L=\lambda+\sum_{n=1}^\infty u_n(\mathbf{t})\lambda^{-n},
\label{form01}
\\&&
M=M_0+M_1,\quad M_0=\sum_{n=0}^\infty t_n L^{n},
\nn \\&&
M_1=\sum_{n=1}^\infty v_n(\mathbf{t})L^{-n}=
\sum_{n=1}^\infty \tilde v_n(\mathbf{t})\lambda^{-n},
\label{form1}
\eea
and $x=t_0$, $(\sum_{-\infty}^{\infty}u_n \lambda^n)_+
=\sum_{n=0}^{\infty}u_n \lambda^n$, $\{L,M\}=L_\lambda M_x-L_x M_\lambda$.
A more standard choice of times for dKP hierarchy corresponds to 
$M_0$=${\sum_{n=0}^\infty (n+1)t_n L^{n}}$, it is easy to transfer to it 
by rescaling of times.

Lax-Sato equations  (\ref{genSato1}) are equivalent to the generating relation
\cite{BDM07,LVB09}
\be
\left(\frac{\d L\wedge \d M}{\{L,M\}}\right)_-=0,
\label{analyticity0}
\ee
where differential takes into account all times $\mathbf{t}$ and 
variable $\lambda$.

A dressing scheme for Manakov-Santini hierarchy can be formulated
in terms of two-component nonlinear Riemann problem on the unit circle $S$
in the complex plane of the variable $\lambda$,
\bea
\begin{aligned}
&L_\text{in}=F_1(L_\text{out},M_\text{out}),
\\
&M_\text{in}=F_2(L_\text{out},M_\text{out}),
\end{aligned}
\label{RiemannMS}
\eea
where the functions 
$L_\text{in}(\lambda,\mathbf{t})$, $M_\text{in}(\lambda,\mathbf{t})$ 
are analytic inside the unit circle,
the functions $L_\text{out}(\lambda,\mathbf{t})$, $M_\text{out}(\lambda,\mathbf{t})$ 
are analytic outside the
unit circle and have an expansion of the form (\ref{form01}), (\ref{form1}).
It is straightforward to demonstrate that  relation (\ref{RiemannMS}) 
implies analyticity of the differential form
$$
\omega=\frac{\d L\wedge \d M}{\{L,M\}}
$$
in the complex plane and generating relation (\ref{analyticity0}), thus defining
a solution of Manakov-Santini hierarchy. Considering a reduction to area-preserving
diffeomorphisms \text{SDiff(2)}, we obtain the dKP hierarchy. 

To obtain a gauge field extension of the hierarchy,
we introduce also a matrix
Riemann-Hilbert problem
\beaa
\Psi_\text{in}=\Psi_\text{out}R(L_\text{out},M_\text{out}),
\eeaa
$\Psi$ is normalised by 1 at infinity and analytic inside
and outside
the unit circle,
\be
\Psi_\text{out}=
1+\sum_{n=1}^\infty \Psi_n(\mathbf{t})\lambda^{-n}
\label{Psiser}
\ee
Expansions of $\Psi$, $L$, $M$ give coefficients for
extended vector fields, $\Psi$ is a wave function. A general
wave function is given by the expression
$\Psi F(L,M)$, $F$ is an arbitrary 
complex-analytic matrix-valued function.
From matrix
RH problem we get analyticity of the
matrix-valued form
\beaa
\Omega=\omega\wedge\d\Psi\cdot \Psi^{-1}
\eeaa
and 
of the functions $\Psi^{-1}X_i\Psi$,
leading to Lax-Sato equations for the series 
$\Psi$ (\ref{Psiser}), $L$ (\ref{form01}), $M$ (\ref{form1}),
defining the evolution of these series,
\begin{align*}
&
\frac{\partial}{\partial t_n}
\begin{pmatrix}
L\\
M
\end{pmatrix}
=V_n(\lambda)
\begin{pmatrix}
L\\
M
\end{pmatrix}
,
\\
&
\frac{\partial}{\partial t_n} {\Psi}=\left(V_n(\lambda)-
((V_n(\lambda)\Psi)\cdot\Psi^{-1})_+\right)\Psi,
\end{align*}
where vector fields $V_n(\lambda)$ are defined 
by formula (\ref{genSato1}).
First flows give exactly extended Lax pair
(\ref{MSLaxext}),
if we identify $y=t_1$, $t=t_2$.
\section{Toda type hierarchy and its 
matrix extension}
\label{Ext1}
Let us recall a picture of the hierarchy 
connected with system (\ref{gen2DTL}) \cite{LVB10}.
A complete set of Lax-Sato equations reads
\be
\begin{aligned}
&\left(
\frac{\partial_{x_n}}{n+1}
-\left(\frac{\l(\e^{(n+1)\Lambda})_\l}
{\{\Lambda,M\}}
\right)^\text{out}_+\p_t
+ 
\left(\frac{(\e^{(n+1)\Lambda})_t}
{\{\Lambda,M\}}
\right)^\text{out}_+
\l\p_\l
\right)
\begin{pmatrix}
\Lambda\\
M
\end{pmatrix}=0,
\\
&
\left(
\frac{\partial_{y_n}}{n+1}
+\left(\frac{\l(\e^{-(n+1)\Lambda})_\l}
{\{\Lambda,M\}}
\right)^\text{in}_-\p_t
- 
\left(\frac{(\e^{-(n+1)\Lambda^-})_t}
{\{\Lambda,M\}}
\right)^\text{in}_-
\l\p_\l
\right)
\begin{pmatrix}
\Lambda\\
M
\end{pmatrix}=0,
\end{aligned}
\label{Hi2}
\ee
where 
the definition of the Poisson bracket is $\{f,g\}=\l(f_\l g_t-f_t g_\l)$,
and we consider formal series
\begin{align}
&\Lambda^\text{out}=\ln\lambda+\sum_{k=1}^{\infty}l^+_k\lambda^{-k},\quad
\Lambda^\text{in}=\ln\lambda+\phi+\sum_{k=1}^{\infty}l^-_k\lambda^{k},
\label{Lform}
\\
&M^\text{out}=M_0^\text{out} + \sum_{k=1}^{\infty}m^+_k\e^{-k\Lambda^{+}},\quad
M^\text{in}=M_0^\text{in} + m_0 +\sum_{k=1}^{\infty}m^-_k\e^{k\Lambda^{-}},
\label{Mform}
\\
&
M_0=t+x\e^\Lambda+y\e^{-\Lambda}+\sum_{k=1}^{\infty}x_k\e^{(k+1)\Lambda}+
\sum_{k=1}^{\infty}y_k\e^{-(k+1)\Lambda},
\nn
\end{align}
where  $\lambda$ is a spectral variable. Usually we
suggest that `out' and `in' components of the series define the functions
outside and inside the unit circle in the complex plane of the variable
$\lambda$ respectively, 
with $\Lambda^\text{in}-\ln\lambda$, $M^\text{in}-M_0^\text{in}$ analytic
in the unit disc, and 
$\Lambda^\text{out}-\ln\lambda$, $M^\text{out}-M_0^\text{out}$ analytic
outside the unit disc and decreasing at infinity.
For a function on the complex plane, having a discontinuity on the unit circle, 
by `in' and `out' components we mean the function inside and outside the unit disc.

A dressing scheme for the  hierarchy (\ref{Hi2})
can be formulated
in terms of the two-component nonlinear Riemann-Hilbert problem on the unit circle $S$,
\bea
\begin{aligned}
&\Lambda^\text{out}=F_1(\Lambda^\text{in},M^\text{in}),
\\
&M^\text{out}=F_2(\Lambda^\text{in},M^\text{in}),
\end{aligned}
\label{Riemann}
\eea
where the functions 
$\Lambda^\text{out}(\l,\mathbf{x},\mathbf{y},t)$, 
$M^\text{out}(\l,\mathbf{x},\mathbf{y},t)$ 
are defined outside the unit circle,
the functions $\Lambda^\text{in}(\l,\mathbf{x},\mathbf{y},t)$, 
$M^\text{in}(\l,\mathbf{x},\mathbf{y},t)$ 
inside the
unit circle by the series of the form (\ref{Lform}), (\ref{Mform}),
with $\Lambda^\text{in}-\ln\lambda$, $M^\text{in}-M_0^\text{in}$ analytic
in the unit disc, and 
$\Lambda^\text{out}-\ln\lambda$, $M^\text{out}-M_0^\text{out}$ analytic
outside the unit disc and decreasing at infinity.

Lax-Sato equations for the times $x=x_1$, $y=y_1$
\beaa
&&
\partial_x
\begin{pmatrix}
\Lambda\\
M
\end{pmatrix}
=\left((\l+ (m_1^+)_t-l_1^+)\partial_t - 
\l l_1^+\partial_\l\right)
\begin{pmatrix}
\Lambda\\
M
\end{pmatrix}
\\
&&
\partial_y
\begin{pmatrix}
\Lambda\\
M
\end{pmatrix}
=\left(\frac{1}{\l}\frac{\mathrm{e}^{-\phi}}{m_t}\partial_t +
\frac{(\mathrm{e}^{-\phi})_t}{m_t}\partial_\l\right)
\begin{pmatrix}
\Lambda\\
M
\end{pmatrix}
,
\eeaa
$m=m_0+t$,
correspond to the Lax pair (\ref{Laxgen2DTL}), where the coefficients 
in the first Lax-Sato equation can be transformed to the form (\ref{Laxgen2DTL})
by taking its expansion at $\lambda=0$, and the system (\ref{gen2DTL}) arises as
a compatibility condition.

To obtain a matrix extension of the hierarchy,
we introduce also a matrix
Riemann-Hilbert problem
\bea
\Psi^\text{out}=
\Psi^\text{in}R(\Lambda^\text{in},M^\text{in}),
\label{RHT}
\eea
$\Psi$ is normalised by 1 at infinity and analytic inside
and outside
the unit circle,
\be
\Psi^\text{out}=
1+\sum_{n=1}^\infty \Psi^+_n(\mathbf{t})\lambda^{-n},
\quad
\Psi^\text{in}=
\sum_{n=0}^\infty \Psi^-_n(\mathbf{t})\lambda^{n}.
\label{Psiser1}
\ee
Lax-Sato equations for the series (\ref{Psiser1})
define the evolution of these series on the background
defined by Lax-Sato equations (\ref{Hi2})
\begin{align*}
&
\frac{\partial_{x_n}}{n+1}
\begin{pmatrix}
\Lambda
\\
M
\end{pmatrix}
=V^+_n(\lambda)
\begin{pmatrix}
\Lambda
\\
M
\end{pmatrix}
\\
&
\frac{\partial_{y_n}}{n+1}
\begin{pmatrix}
\Lambda
\\
M
\end{pmatrix}
=V^-_n(\lambda)
\begin{pmatrix}
\Lambda
\\
M
\end{pmatrix}
,
\\
&
\frac{\partial_{x_n}}{n+1}
{\Psi}=\left(V^+_n(\lambda)-
((V^+_n(\lambda)\Psi)\cdot\Psi^{-1})^\text{out}_+\right)
\Psi,
\\
&
\frac{\partial_{y_n}}{n+1}
{\Psi}=\left(V^-_n(\lambda)-
((V^-_n(\lambda)\Psi)\cdot\Psi^{-1})^\text{in}_-\right)
\Psi,
\end{align*}
where vector fields $V^+_n(\lambda)$, 
$V^-_n(\lambda)$ are defined 
by formula (\ref{Hi2}) and have coefficients  
polynomial respectively
in $\lambda$ and $\lambda^{-1}$.
\subsection*{Acknowledgements}
This work was performed in the framework of
State assignment topic 0033-2019-0006 (Integrable
systems of mathematical physics).

\end{document}